\documentclass[12pt, draftclsnofoot, onecolumn]{IEEEtran}
\usepackage{amsmath,amsthm,tikz,color}
\usepackage{amssymb}%
\usepackage{mathrsfs}
\newtheorem{theorem}{Theorem}
\usepackage{mathtools}
\DeclarePairedDelimiter{\abs}{\lvert}{\rvert}
\usepackage[utf8]{inputenc} 
\usepackage[T1]{fontenc}
\ifCLASSINFOpdf
\else 
\fi  
\usepackage{tabularx,hhline}
\usepackage{dashbox}
\usepackage{color}
\usepackage{multirow}
\usepackage{ifpdf}
\usepackage{hhline}
\usepackage{array}
\usepackage{cite}
\usepackage{dblfloatfix}
\usepackage{xcolor}

\usepackage{flushend}
\usepackage{setspace}
\usepackage{graphicx}
\usepackage{subcaption}
\usepackage[font=small,labelfont=bf]{caption}
\begin{document}
\title{Wireless Network Reliability Analysis for Arbitrary Network Topologies}

\author{Semiha~Tedik Basaran$^{\dag}$,~\IEEEmembership{Student Member, IEEE},\\
Gunes~Karabulut Kurt$^{\dag}$,~\IEEEmembership{Senior Member, IEEE},\\
and Frank R. Kschischang$^\ddag$, ~\IEEEmembership{Fellow, IEEE}
\thanks{$^\dag$
S. Tedik Basaran and G. Karabulut Kurt are with the Department of
Communications and Electronics Engineering, Istanbul Technical University,
Turkey, e-mail:  \texttt{\{tedik, gkurt\}@itu.edu.tr}.\newline
$^\ddag$
Frank R. Kschischang is with the Edward S. Rogers Sr.\ Department of
Electrical and Computer Engineering, University of Toronto, Canada, e-mail:
\texttt{frank@ece.utoronto.ca}.}}

\maketitle
	
\begin{abstract}
The outage performance of wireless networks with unstructured network
topologies is investigated.  The network reliability perspective of graph
theory is used to obtain the network outage polynomial of generalized wireless
networks by enumerating paths and cut-sets of its graph representation for both
uncorrelated and correlated wireless channels.  A relation is established
between the max-flow min-cut theorem and key communication performance
indicators. The diversity order is equal to the size of the minimum cut-set between
source and destination, and the coding gain is the number of cut-sets with size
equal to the minimum cut.  An ergodic capacity analysis of networks with
arbitrary topologies based on the network outage polynomial is also presented.
Numerical results are used to illustrate the technical definitions
and verify the derivations.
\end{abstract}
	
	
\IEEEpeerreviewmaketitle

\section{Introduction}

\IEEEPARstart{N}{etwork} topologies in wireless environments are generally
dynamic in nature, as the connectivity between nodes is determined according to
their time-varying link signal-to-noise ratio (SNR) value. Channel impairments
such as fading and path loss make it essential to monitor the quality of each
link.

Based on the corresponding SNR value, the outage status can be determined and
used as a performance indicator for each link. In case of an outage (where the
SNR value falls below a certain threshold), two nodes are deemed to be
disconnected; otherwise, they remain connected. Outage probability is a
convenient measure of communication system performance\cite{french1979effect}.
Here we investigate \textit{network outage}, i.e., the outage probability of
communication between a source node and a terminal node over a network of relay
nodes.  Assuming links are in outage independently,  network outage can be
measured by using individual link outage probabilities. The behavior of the
outage probability in the high SNR regime also gives an intuitive understanding
of performance limits of the network \cite{proakis}. In
\cite{Giannakis_diversity_coding_gain}, high SNR error performance of any
communication network (coded or uncoded) is represented by diversity and coding
gains. \textit{Diversity gain} is a measure of the number of
independent copies of the transmitted signal captured by the receiver
\cite{Zheng2003}, while \textit{coding gain}
represents the difference in the outage
probability curve relative to a benchmark performance in the high SNR region
\cite{Giannakis_diversity_coding_gain}.

In this paper, we show that the diversity gain and the coding gain between a
source node and a terminal node can be determined through the \textit{network
outage polynomial}.  This approach dates back at least to Shannon and Moore
\cite{MOORE1956191}, who provide a reliability analysis of relay-aided systems
by considering the unreliability probabilities of relay nodes.  It is proven
that the end-to-end reliability of a given network can be increased through
these unreliable relay nodes.  When a sufficient number of relay nodes is used,
the probability of network unreliability approaches zero \cite{MOORE1956191}.
However, transforming a complex network into equivalent series-parallel
projection may not always be possible. When the series-parallel representation
of a given network is not available, the reliability analysis of generalized
networks becomes more difficult. There are various methods proposed to
calculate network reliability, such as state enumeration, factorizing, path
enumeration, and cut-set enumeration \cite{bondy1976graph,wing1964analysis,
reliability,moskowitz1958analysis}. 

Although network reliability is a well-studied subject, its extension to
wireless networks it is still relatively unexplored
\cite{gupta1999critical,bettstetter2005connectivity,hekmat2006connectivity,
agrawal2009correlated,lu2014reliability}.
As the popularity of wireless communication systems increases when compared to
their wired counterparts in many different areas, the reliability analysis of
wireless communications becomes more important, yet challenging, as wireless
links are more prone to errors and erasures.  Firstly, an unrealistic
deterministic channel model is used when investigating the interference effect
of the wireless channels \cite{gupta1999critical}. The reliability analysis of
wireless multi-hop networks is conducted regarding shadowing effect of the
wireless channel in  \cite{bettstetter2005connectivity,hekmat2006connectivity}.
Both \cite{bettstetter2005connectivity,hekmat2006connectivity} do not consider
the correlation effect of shadowing and this gap is filled by
\cite{agrawal2009correlated}. The reliability analysis of wireless multi-hop
networks, which proposes a mathematical model to represent the network
reliability of correlated shadowing wireless channel, is given in
\cite{lu2014reliability}. 

In \cite{reliability}, path-enumeration and cut-set enumeration methods are
used to calculate network reliability of generalized schemes.  An algorithm
based on a path-enumeration method is presented in \cite{rai1991computer} to
determine the reliability of telecommunication networks from the capacity of
the networks by considering different link capacities. In \cite{4505532}, a
path-based algorithm with a reduced computational time is modeled to obtain
network reliability of wired communications. Instead of considering all
cut-sets of the network, some cut-sets that have as many elements as the size
of the size of the minimum cut are used to obtain an approximate network
reliability expression with reduced complexity \cite{jensen1969algorithm}.
Hence, a lower bound for network reliability is attained by providing a more
practical algorithm. 

The
network outage polynomial gives the probability
that the network has zero instantaneous capacity.
The investigation of network capacity is an 
attractive problem since the maximum capacity of any network is restricted
by the size of the minimum cut of the graph. Hence, the ergodic capacity
of any network can be calculated by using zero-to-$m$ capacity polynomials, 
where the $i$th polynomial gives the probability that the network
has instantaneous capacity $i$.
In the literature, there are some works about the calculation of the capacity
polynomials that determine the value of the maximum flow of arbitrary networks
with random capacity edges by utilizing subset decomposition method 
\cite{doulliez1972transportation, grimmett1982flow}. 
In \cite{doulliez1972transportation}, a subspace decomposition principle 
is used to determine the value of the maximum flow of arbitrary networks 
with random capacity edges. The value of maximum flow analysis of arbitrary
networks with random edge capacities is conducted in \cite{grimmett1982flow},
based upon Bernoulli statistics.

The aforementioned works have focused on obtaining only network reliability
expressions. On the other hand, these works do not introduce any fundamental
performance analysis. In this work, the essential goal is to obtain performance
limits of an arbitrary network topology comprised of links that are prone to
errors and erasures. The main contributions of this work can be listed as
follows:
\begin{itemize}
 \item We establish a framework to calculate the \textit{network outage
polynomial}, as a tool to obtain network outage performance of communication
networks.
 \item We determine the network outage polynomial of some simple
directed
networks, in both correlated channels and uncorrelated channels. Three methods,
namely the path-enumeration method, the cut-set enumeration method, and the
terminal reliability based method are proposed.
 \item  We extract the diversity order and the coding gain of a wireless network
for arbitrary topology based on its graph properties. 
 \item We establish a relationship between the max-flow min-cut theorem of graph
theory and the diversity gain definition and show that the diversity order
corresponds to the size of the minimum cut of the wireless network graph.  We
also prove that the coding gain is equal to the number of cut-sets which have
the size of the minimum cut, and also be easily determined from the network
graph. 
 \item We provide the ergodic capacity analysis of networks in terms of
individual link outage probability. Hence, an upper bound for the achievable
transmission rate is determined. 
\end{itemize}
Using this analysis, optimization of resource utilization can be realized
thanks to the information about the diversity order and  the ergodic capacity
of any topology in wireless networks. For example, efficient multiple access
schemes can be obtained by considering user demands and network limitations
(the diversity order and the ergodic capacity).	

The rest of the paper is organized as follows. Firstly, in Section II, methods
for the calculation of outage polynomials of wireless networks are given. In
Section III, diversity order analysis and ergodic capacity derivations are
presented.  In Section IV, to demonstrate the validity of theoretical results, 
numerical results are presented. Finally, the paper concludes with a
summary of the findings and suggestions for future work in Section V.

\section{Outage Polynomials of Wireless Networks}

Graph representations of communication systems are frequently used to analyze
system performance; hence, key graph theory concepts can often be matched with
the elements of communication systems.  In literature pertaining to wired
networks, the link outage probability is generally ignored since links are
generally highly reliable.  Thus, for wired networks, the connections between
nodes can be represented by deterministic edges.  The links in wireless
channels, on the other hand, are subject to random SNR values, and so the
connections between nodes must be modeled probabilistically.

We model a communications network $\mathcal{N} = (\mathcal{V},\mathcal{E},s,t)$
as a directed acyclic network compromising of a finite vertex set $\mathcal{V}$
of communication nodes, a multi-set of $n$ directed edges
$\mathcal{E} = \{
e_1, e_2, \ldots, e_n \} \subseteq \mathcal{V} \times \mathcal{V}$ representing
communication
links between nodes, a designated source vertex $s$ and a designated terminal
vertex $t$ where $s, t \in V, \; s\neq t$.  An edge $e$
from vertex $v$ to vertex $w$ is denoted as $v \to w$.

A directed path in $\mathcal{N}$ from $s$ to $t$ is a sequence of edges $(v_0
\to v_1),(v_1 \to v_2),\ldots,(v_{\ell-1}\to v_{\ell})$ with $v_0 = s$ and
$v_{\ell}=t$. We suppose that there are $g$ distinct paths $\mathcal{P}_1,
\ldots, \mathcal{P}_g$ in $\mathcal{N}$ from $s$ to $t$.  Nodes $s$ and $t$ are
said to be \emph{connected} if $g \geq 1$.

A subset $\mathcal{C} \subseteq \mathcal{E}$ of edges whose
removal from the network
disconnects $s$ and $t$ is called an \emph{$s$-$t$-separating
cut}, or simply a \emph{cut-set}.  We suppose that
there are $k$ distinct cut-sets $\mathcal{C}_1, \ldots, \mathcal{C}_k$;
the collection of all cut-sets is denoted as $\mathcal{K}$.

A cut-set $\mathcal{C} \in \mathcal{K}$ is called \emph{minimal} if no proper
subset of $\mathcal{C}$ is itself a cut-set.  The collection of all minimal
cut-sets is denoted as $\mathcal{L}$.  A cut-set $\mathcal{C} \in \mathcal{K}$
is called a \emph{minimum cut-set} if it is a cut-set of minimum possible size,
i.e., having the least number of edges among all cut-sets.  The collection of
all minimum cut-sets is denoted as $\mathcal{M}$, and the size of any minimum
cut-set is denoted as $m$.  Although each minimum cut-set is certainly a
minimal cut-set, the converse is not true in general, thus $\mathcal{M}
\subseteq \mathcal{L} \subseteq \mathcal{K}$.

Network outage is a convenient measure of a communication system's performance,
as the overall system performance can be obtained using individual outage
probabilities of the links in the system. To enable communication between a
source node $s$ and a terminal node $t$, there must be at least one path from
$s$ to $t$.  Hence, we can obtain an overall performance result by considering
individual link outages.  The network outage polynomial concept,
which has been
proposed for switching networks
\cite{MOORE1956191,Co87},
is also suitable as performance observation tool
for wireless communication.
Network outage is random due to individual link
outages.  In order to obtain the network outage polynomial for an arbitrary
topology, we use three different methods: path enumeration, cut-set
enumeration, and reliability polynomial calculation.  The required method can
be selected to realize the target aim, as detailed below.

In the following,
we consider the network at a given time instant, and denote
by $p_j$ the probability that link $e_j$ is in outage at that instant.
For example, if the wireless channel gain
$\abs{h_j}$ has a Rayleigh distribution (a frequent assumption
in the wireless communication literature), then the outage
probability of $e_j$ is equal to
\begin{equation*}
p_j=1-\exp \left(-\gamma_j^{-1} \right),
\label{eqn1}
\end{equation*}
where ${\gamma}_j$ represents the average SNR of the link $e_j$ \cite{Zheng2003}.

Link outages induce a random subgraph of $\mathcal{N}$, called
the \emph{residual network},
with edges that are in outage removed.  In the residual network, it
may happen that $s$ and $t$ are not connected.
The \emph{network outage polynomial}, which gives
the probability that no path exists between $s$ and $t$ in
the residual network, is
then formally a polynomial function of $p_1, \ldots, p_n$,
denoted as $O(p_1,\ldots,p_n)$.

Throughout this paper, for any positive integer $\ell$, we will denote the
set $\{ 1, 2, \ldots, \ell \}$ as $[ \ell ]$.

\subsection{Network Outage Polynomial Calculation Based on Path Enumeration}

Firstly, we investigate the path enumeration method to obtain the network
outage polynomial.
We suppose that 
the edges comprising a path
$\mathcal{P}_r$ in $\mathcal{N}$ from $s$ to $t$ are indexed by the set
$\mathscr{P}_r \subseteq [n]$, i.e.,
$\mathcal{P}_r = \{ e_j : j \in \mathscr{P}_r \}$, $r \in [g]$.

Let $Q_r$ denote the event that path $\mathcal{P}_r$
is available, i.e., that none of its links are in outage.
The outage probability of the network is then given by
\[
O(p_1, \ldots, p_n) = 1 - \Pr[ Q_1 \cup Q_2 \cup \cdots \cup Q_g ].
\]
By the principle of inclusion-exclusion \cite{andrews1994}, we have
\begin{align}
\Pr[ Q_1 \cup Q_2 \cup \cdots \cup Q_g ] &=
 \sum_{i_1 \in [g]} \Pr[ Q_{i_1} ]
 -\sum_{\substack{i_1, i_2  \in [g] \\ i_1 \neq i_2}}
    \Pr[Q_{i_1} \cap Q_{i_2}] + \cdots  \nonumber\\
 & + (-1)^{\beta-1} \sum_{\substack{i_1, i_2, \ldots, i_{\beta} \in [g]\\
   i_1, i_2, \ldots, i_{\beta} \text{ distinct}}}
   \Pr[ Q_{i_1} \cap Q_{i_2} \cap \cdots \cap  Q_{i_{\beta}}] + \cdots \nonumber\\
&  +  (-1)^{g-1}  \Pr[ Q_1 \cap Q_2 \cap \cdots \cap Q_g ].
\label{eqn3}
\end{align}
Assuming that individual links are in outage (or not) independently, 
we have
\begin{equation}
\Pr [Q_{i_1} \cap Q_{i_2} \cap \cdots \cap Q_{i_{\beta}}]
= \prod_{j \in \mathscr{P}_{i_1} \cup \mathscr{P}_{i_2} \cup \cdots
\cup \mathscr{P}_{i_\beta}} (1 - p_j).
\label{eqn2}
\end{equation}

\subsection{Network Outage Polynomial Calculation Based on Cut-Set Enumeration}

The network outage polynomial of an arbitrary network can
also be calculated by
enumerating cut-sets of the network, which is dual to the
process of path enumeration.
If the 
edges of any cut-set are all in outage, the network is in outage.

We suppose that
the edges comprising a cut-set
$\mathcal{C}_r$ are indexed by the set
$\mathscr{C}_r \subseteq [n]$, i.e.,
$\mathcal{C}_r = \{ e_j : j \in \mathscr{C}_r \}$, $r \in [k]$.

Let $D_r$ denote the event that cut-set $\mathcal{C}_r$
is active, i.e., that all of its links are in outage.
The outage probability of the network is then given by
\[
O(p_1, \ldots, p_n) = \Pr[ D_1 \cup D_2 \cup \cdots \cup D_k ].
\]
Again by the principle of inclusion-exclusion we have
\begin{align}
\Pr[ D_1 \cup D_2 \cup \cdots \cup D_k ] &=
 \sum_{i_1 \in [k]} \Pr[ D_{i_1} ]
 -\sum_{\substack{i_1, i_2  \in [k] \\ i_1 \neq i_2}}
    \Pr[D_{i_1} \cap D_{i_2}] + \cdots  \nonumber \\
 & + (-1)^{\beta-1} \sum_{\substack{i_1, i_2, \ldots, i_{\beta} \in [k] \\
   i_1, i_2, \ldots, i_{\beta} \text{ distinct}}}
   \Pr[ D_{i_1} \cap D_{i_2} \cap \cdots \cap  D_{i_{\beta}}] + \cdots \nonumber\\
&  +  (-1)^{k-1}  \Pr[ D_1 \cap D_2 \cap \cdots \cap D_k],
\end{align}
where
\begin{equation}
\Pr [D_{i_1} \cap D_{i_2} \cap \cdots \cap D_{i_{\beta}}]
= \prod_{j \in \mathscr{C}_{i_1} \cup \mathscr{C}_{i_1} \cup \cdots
\cup \mathscr{C}_{i_\beta}} p_j.
\end{equation}

\subsection{Network Outage Polynomial Calculation Based on Two-Terminal Polynomial}

Finally, we derive the network outage polynomial expressions of a network based
on the reliability polynomial concept \cite{Co87}, which is a useful function
to reflect the performance of a network.

Consider, for any cut-set $\mathcal{C}_r$, $r \in [k]$, the event $E_r$ that
all the edges of $\mathcal{C}_r$ are in outage while all \emph{other} edges of
the network are \emph{not} in outage.  Since $E_{r}$ is disjoint from $E_{s}$
when $r \neq s$, we have
\[
O(p_1, \ldots, p_n) = \Pr \left[ \bigcup_{r \in [k]} E_r \right]
= \sum_{r \in [k]} \Pr[ E_r ].
\]
Again assuming that
individual links are in outage (or not) independently, 
we have
\[
\Pr[ E_r] = \prod_{j \in \mathscr{C}_r} p_j \cdot \prod_{i \in [k] \setminus \mathscr{C}_r} (1-p_i).
\]

In the special case where $p_j = p$ for all $j \in [n]$, we
have
\[
\Pr [E_r] = p^{|\mathcal{C}_r|} (1 - p)^{n - |\mathcal{C}_r|}.
\]
Writing $O(p)$ for the outage polynomial in this case, we get
\begin{align}
O(p) & = \sum_{r \in [k]} p^{|\mathcal{C}_r|} (1 - p)^{n - |\mathcal{C}_r|}
     = \sum_{i=m}^n A_i p^i (1-p)^{n-i} \nonumber\\
& = (1-p)^n A\left( \frac{p}{1-p} \right),
\label{eqn16}	
\end{align}
where
\begin{equation}
A(x)  = \sum_{\mathcal{C} \in \mathcal{K}} x^{|\mathcal{C}|} = 
A_m x^m + A_{m+1} x^{m+1} + \cdots + A_{n} x^n,
\label{eqn17}
\end{equation}
and where the coefficient $A_i$ of $x^i$ enumerates the number of cut-sets of size $i$. 

It can be deduced from the minimum cut-set definition that $A_m$ is equal to
the number of distinct minimum cut-sets and $A_m \neq 0$.  In addition, $A_n$
is equal to 1.  The outage polynomial can be also expressed in terms of the
reliability polynomial associated with the $\text{Conn}_2(\mathcal{N})$ $s$-$t$
connectedness problem,  $O(p) = 1 - \text{Rel}(\mathcal{N},1-p)$
\cite[Sec.~1.2]{Co87}.

The computational complexity of the outage polynomial depends on
the determination of $\mathcal{K}$. The complexity per cut is given as 
$\mathcal{O}(n)$ in \cite{provan1996p}. Hence, the enumeration of cut-sets 
can be found as $\mathcal{O}(kn)$ where the number of all cut-sets 
($k=|\mathcal{K}|$) depends on the size of $\mathcal{N}$ \cite{provan1996p,ball1995network}.
\subsection{Bounds on the Outage Polynomial}

We may write some simple bounds on the outage polynomial as follows.

Firstly, if we use the inequality of $(1-p) \leq 1$ in (\ref{eqn16}),
we get
\begin{equation}
O(p) \leq \sum_{i=m}^n A_i p^i = A(p)
\label{eqn18}
\end{equation}
To derive another upper bound expression, we can use the fact that every cut-set
must contain a minimal cut-set.  Since the probability
that edges of a cut-set $\mathcal{C}$ are in outage is
$p^{|C|}$, we get that
\begin{equation}
O(p) \leq \sum_{C \in \mathcal{L}} p^{|C|}.
\label{eqn19}
\end{equation}
We also have the lower bound 
\begin{equation}
O(p) \geq A_m p^m(1-p)^{n-m}
\label{eqn20}
\end{equation}
which is obtained by retaining just the first term in the expansion
$O(p) = \sum_{i=m}^n A_i p^i (1-p)^{n-i}$.

\subsection{Presence of Correlated Channels}

In the previous subsections, we have assume that the state of each link is
independent of the others. This assumption may be unrealistic in many
situations (e.g., multi-antenna systems) because of spatial correlation. The
correlated channel case needs to be considered to determine the limitations of
the wireless networks. 

We adopt a simple correlation model, as follows.
Firstly, the set $\mathcal{E}$ of links is partitioned into
disjoint nonempty subsets,
$\mathcal{B}_1$, $\mathcal{B}_2$, \ldots, $\mathcal{B}_f$,
so that
\[
\bigcup_{i=1}^f \mathcal{B}_i = \mathcal{E} \quad \text{and} \quad
i \neq j \text{ implies } \mathcal{B}_i \cap \mathcal{B}_j = \emptyset.
\]
To subset $\mathcal{B}_i$ is associated a Bernoulli ($\{0,1\}$-valued)
random variable
$S_i$, with $\Pr[S_i=1]=\rho$.  If $S_i = 1$, the link states
(in outage or not) for all links in $\mathcal{B}_i$
are chosen to be equal,
while if $S_i = 0$, the link states for the
links in $\mathcal{B}_i$ are chosen independently at random.
Suppose that $\mathcal{B}_i$ has size $|\mathcal{B}_i|=x$,
and let $\mathcal{S}_i$ be any subset of $\mathcal{B}_i$ of
size $|\mathcal{S}_i|=y$, where $0 \leq y \leq x$.
Then the probability $p_o(x,y)$ that the links of $\mathcal{S}_i$ are
in outage while the links of $\mathcal{B}_i \setminus \mathcal{S}_i$
are \emph{not} in outage is given as
\begin{equation}
p_o(x,y) = 
\begin{cases}
\rho(1-p) + (1-\rho)(1-p)^x & \text{if } y=0 \\
\rho p + (1-\rho)p^x & \text{if }y = x\\
(1-\rho) p^y (1-p)^{x-y} & \text{otherwise}.
\end{cases}
\label{eqn:correlation}
\end{equation}
We assume that the random variables $S_1, \ldots, S_f$ are mutually
independent.
Note that the previously considered case (of independent link-states)
is obtained by considering $\rho=0$, or, equivalently,
by partitioning $\mathcal{E}$
into singleton sets where $|\mathcal{B}_i| = 1$ for all $i$.

Now, given any subset $\mathcal{C} \subset \mathcal{E}$ of edges
(e.g., a cut-set), the probability that all edges of $\mathcal{C}$
are in outage while all edges in $\mathcal{E} \setminus \mathcal{C}$
are \emph{not} in outage is given by
\[
\prod_{i=1}^f p_o(|\mathcal{B}_i|,|\mathcal{C} \cap \mathcal{B}_i|).
\]
Thus the network outage polynomial is obtained as
\begin{equation}
O(p) =  \sum_{\mathcal{C} \in \mathcal{K}}
\prod_{i=1}^f p_o(|\mathcal{B}_i|,|\mathcal{C} \cap \mathcal{B}_i|).
\label{eqn:correlatedoutage}
\end{equation}

\section{Diversity Order and Ergodic Capacity Analyses for Arbitrary Network Topologies}

In this section, performance limitations of an arbitrary network are determined
via the outage polynomial. Firstly, expressions for diversity gain and coding
gain are derived. Secondly, the ergodic capacity is considered.

\subsection{Diversity Order Analysis}

In order to provide further insight into the obtained outage probability
expression, an asymptotic expression of outage probability is derived.
The network is in outage if there is no defined path between a source and
terminal nodes. Coding and diversity gains can represent the network outage
probability in the limit as $p \rightarrow 0$, referred to as the
high SNR regime.
The high
SNR  performance of any system determines the performance limits of a wireless
network. In the high SNR regime, the outage probability expression of an arbitrary given
network is given as
\[
O(p) \approx \alpha \gamma^{-d},
\]
where $d$, the \emph{diversity gain}, measures the number of
independent copies of the transmitted signal that are received
at the terminal node, and
where  $\alpha$, the \emph{coding gain} (usually expressed
on a decibel scale), is a
measure of the performance difference between the given system and
a baseline system having
$O(p) \approx \gamma^{-d}$ \cite{wang2003simple}.

For the purposes of the following theorem,
we say that two functions $f(p)$ and $g(p)$
are asymptotically equal, written $f(p) \sim g(p)$,
if
\[
\lim_{p \to 0} \frac{f(p)}{g(p)} = 1.
\]

\begin{theorem}
In a network with outage polynomial
$O(p) = \sum_{i=m}^n A_i p^i (1-p)^{n-i}$,
\[
O(p) \sim A_m p^m.
\]
Thus the diversity order of such a network is equal
to the size of a minimum cut-set, i.e., $d=m$, and
the coding gain is equal to the number of distinct
minimum cut-sets, i.e., $\alpha=A_m$.
\end{theorem}

\begin{proof}
We have
\begin{align}
\lim_{p \rightarrow 0} \frac{O(p)}{A_m p^m} 
&=\lim_{p \rightarrow 0} \frac{A_m p^m (1-p)^{n-m} 
+ A_{m+1} p^{m+1} (1-p)^{n-m-1} + \cdots + p^n}{A_m p^m}\nonumber \\
&= \lim_{p \rightarrow 0} (1-p)^{n-m} + \lim_{p \rightarrow 0} 
\frac{A_{m+1}}{A_m} p (1-p)^{n-m-1} + \cdots + \lim_{p\rightarrow 0}
\frac{1}{A_m}p^{n-m}\nonumber \\
&= 1.
\label{eqn27}
\end{align}
\end{proof}
The value of maximum flow (the size of the minimum cut) can be calculated by
enumerating the number of cut-sets in a dual manner for unit capacity graphs.
For dense network graphs, the Ford-Fulkerson algorithm can be used to determine
the size of the minimum cut value \cite{ford1956maximal}.

It is obvious that adding new edges to a network cannot reduce the size of any
cut-sets.  If newly added edges (e.g., a line-of-sight edge) provide a new
edge-disjoint path from $s$ to $t$, then the cardinality of all cut-sets,
and hence the diversity order of the network,
increases by one.

\subsection{Ergodic Network Capacity}

Suppose now that each network link (when not in outage) provides
unit transmission capacity.  It is well
known, e.g., \cite{koetter03},  that the instantaneous
$s$-$t$ unicast capacity $C$ is equal to
the size of the minimum $s$-$t$-separating cut in the network
subgraph induced by the links that are \emph{not} in outage; this
transmission rate can be achieved by routing information along edge-disjoint
paths between $s$ and $t$ (which, by Menger's Theorem, exist
in sufficient number).
As the link-state is random, the instantaneous capacity $C$
is a random variable.  Indeed, the outage polynomial $O(p)$
gives the probability that $C=0$.  It is also clear that
$C$ is bounded by $m$, the minimum cut-set size.
As $C$ takes integer values in a bounded set, it has
a well-defined expected value, called the \emph{ergodic network
capacity}.

For $i \in \{ 0, 1, \ldots, m \}$, the event
$C=i$ arises when \emph{all} minimal cut-sets $\mathcal{C} \in \mathcal{L}$
contain at least $i$ links not in outage, and \emph{at least one}
of these cut-sets contains exactly $i$ links not in outage.
In other words, $C=i$ arises when the minimum number of non-outage
links among minimal cut-sets is equal to $i$.
More precisely,
let $\mathcal{E}'$ denote the set of edges \emph{not} in outage
at a given time instant.  For any minimal cut $\mathcal{C} \in \mathcal{L}$,
let
\begin{equation}
\delta_i(\mathcal{C})= \begin{cases}
0 & {\lvert \mathcal{C} \cap \mathcal{E}' \rvert < i} \\
1 & {{\lvert \mathcal{C} \cap \mathcal{E}' \rvert \geq i}}
\end{cases}
\label{eqn13}
\end{equation}
be the function that indicates whether $\mathcal{C}$ contains
at least $i$ edges not in outage.
The event $C=i$ then arises if
\begin{align}
\forall \mathcal{C} \in \mathcal{L} (\delta(\mathcal{C}) = 1) \label{eqn14}\\
\min_{\mathcal{C} \in \mathcal{L}} \lvert \mathcal{C} \cap \mathcal{E}' \rvert = i.
\label{eqn15}
\end{align}

For every $i$, the probability that $C=i$ is given
by some polynomial $C_i(p)$.
The ergodic capacity can then be obtained, in terms of $p$,
as 
\begin{equation}
E[C](p) = \sum_{i=0}^m i C_i(p){.}
\label{ergodic}
\end{equation}
When the minimal cut sets $\mathcal{C} \in \mathcal{L}$
are disjoint,
the $i$th capacity polynomial can be calculated as follows.
For any minimal cut $\mathcal{C} \in \mathcal{L}$ of size
$|\mathcal{C}|$, let
$q(i,|\mathcal{C}|,p)$ denote the probability that $\mathcal{C}$
contains \emph{at least} $i$ links not in outage;  thus
\[
q(i,|\mathcal{C}|,p) = \sum_{j \geq i} \binom{ |\mathcal{C}|}{j}
p^{|\mathcal{C}|-j}(1-p)^j.
\]
The probability that \emph{every} minimal cut contains $i$ or more
links not in outage is then given as
$\prod_{\mathcal{C} \in \mathcal{L}} q(i,|\mathcal{C}|,p)$.
The probability that $C=i$ is then given as
the probability that \emph{every} minimal cut contains $i$ or
more links in non-outage but not every minimal cut contains $i+1$
more links in non-outage, namely
\[
C_i(p) =
\prod_{\mathcal{C} \in \mathcal{L}} q(i,|\mathcal{C}|,p)-
\prod_{\mathcal{C} \in \mathcal{L}} q(i+1,|\mathcal{C}|,p).
\]

\begin{figure}	
	\centering
	\begin{subfigure}{0.4\textwidth}
		\centering
		\begin{tikzpicture}[>=stealth,vertex/.style=
		{circle,draw,minimum size=3.5ex},x=3cm,y=3cm]
		\node[vertex] (A) at (0,0) {$s$};
		\node[vertex] (B) at (0.8,0) {};
		\node[vertex] (C) at (1.6,0) {$t$};
		\draw[->] (A) -- node[above] {$e_1$} (B);
		\draw[->] (B) to[out=30,in=150] node[above] {$e_2$} (C);
		\draw[->] (B) to[out=-30,in=210] node[below] {$e_3$} (C);
		\end{tikzpicture}    
	\caption{$\mathcal{N}_1$}
		\label{fig:N1}
	\end{subfigure}
	\begin{subfigure}{0.4\textwidth}
		\centering 
		\begin{tikzpicture}[>=stealth,vertex/.style=
		{circle,draw,minimum size=3.5ex},x=3cm,y=3cm]
		\node[vertex] (A) at (0,0) {$s$};
		\node[vertex] (B) at (0.8,0) {};
		\node[vertex] (C) at (1.6,0) {$t$};
		\draw[->] (A) to[out=30,in=150] node[above] {$e_1$} (B);
		\draw[->] (A) to[out=-30,in=210] node[below] {$e_2$} (B);
		\draw[->] (B) to[out=30,in=150] node[above] {$e_3$} (C);
		\draw[->] (B) to[out=-30,in=210] node[below] {$e_4$} (C);
		\end{tikzpicture}
	\caption{$\mathcal{N}_2$}
		\label{fig:N2}
	\end{subfigure}
	\begin{subfigure}{0.4\textwidth}
		\centering\vspace{15pt}
		\begin{tikzpicture}[>=stealth,vertex/.style=
		{circle,draw,minimum size=3.5ex},x=3cm,y=3cm]
		\node[vertex] (A) at (0,0) {$s$};
		\node[vertex] (B) at (0.8,0) {};
		\node[vertex] (C) at (1.6,0) {$t$};
		\draw[->] (A) -- node[above] {$e_1$} (B);
		\draw[->] (B) -- node[above] {$e_2$} (C);
		\end{tikzpicture}
		\vspace{35pt}
	\caption{$\mathcal{N}_3$}
		\label{fig:N3}
	\end{subfigure}
	\begin{subfigure}{0.4\textwidth}
	\centering \vspace{20pt}
	\begin{tikzpicture}[>=stealth,vertex/.style=
	{circle,draw,minimum size=3.5ex},x=3cm,y=3cm]
	\node[vertex] (A) at (0,0) {$s$};
	\node[vertex] (B) at (0.4,-0.5) {};
	\node[vertex] (C) at (0.8,0) {};
	\node[vertex] (D) at (1.2,-0.5) {};
	\node[vertex] (E) at (1.6,0) {$t$};
	\draw[->] (A) -- node[above] {$e_2$} (B);
	\draw[->] (A) -- node[above] {$e_1$} (C);
	\draw[->] (B) -- node[above] {$e_3$} (D);
	\draw[->] (C) -- node[above] {$e_5$} (E);
	\draw[->] (C) -- node[above] {$e_4$} (D);
	\draw[->] (D) -- node[above] {$e_6$} (E);
	\end{tikzpicture}
	\caption{$\mathcal{N}_4$}
	\label{fig:N4}
\end{subfigure}
	\begin{subfigure}{0.4\textwidth}
		\centering \vspace{20pt}
		\begin{tikzpicture}[>=stealth,vertex/.style=
		{circle,draw,minimum size=3.5ex},x=3cm,y=3cm]
		\node[vertex] (A) at (0,0) {$s$};
		\node[vertex] (B) at (0.8,0.5) {};
		\node[vertex] (C) at (0.8,-0.5) {};
		\node[vertex] (D) at (1.6,0) {$t$};
		\draw[->] (A) -- node[above] {$e_1$} (B);
		\draw[->] (A) -- node[below] {$e_2$} (C);
		\draw[->] (B) -- node[above] {$e_3$} (D);
		\draw[->] (C) -- node[below] {$e_4$} (D);
		\end{tikzpicture}
		\caption{$\mathcal{N}_5$}
		\label{fig:N5}
	\end{subfigure}
	\begin{subfigure}{0.4\textwidth}
		\centering \vspace{0.5cm}
		\begin{tikzpicture}[>=stealth,vertex/.style=
		{circle,draw,minimum size=3.5ex},x=3cm,y=3cm]
		\node[vertex] (A) at (0,0) {$s$};
		\node[vertex] (B) at (0.8,0.5) {};
		\node[vertex] (C) at (0.8,-0.5) {};
		\node[vertex] (D) at (1.6,0) {$t$};
		\draw[->] (A) -- node[above] {$e_1$} (B);
		\draw[->] (A) -- node[above] {$e_2$} (D);
		\draw[->] (A) -- node[below] {$e_3$} (C);
		\draw[->] (B) -- node[above] {$e_4$} (D);
		\draw[->] (C) -- node[below] {$e_5$} (D);
		\end{tikzpicture}
		\caption{$\mathcal{N}_6$}
		\label{fig:N6}
	\end{subfigure}
	\caption{There are three example networks denoted by $\mathcal{N}_1$, 
		      $\mathcal{N}_2$, $\mathcal{N}_3$, $\mathcal{N}_4$, $\mathcal{N}_5$, 
		      and $\mathcal{N}_6$ which are presented in (a), (b), (c), (d), and (e), 
		      respectively. $\mathcal{N}_1$, $\mathcal{N}_2$, and $\mathcal{N}_3$ 
		      respectively have 3 edges, 4 edges, 2 edges, 6 edges, 4 edges and 5 edges.}
	\label{fig:all_graphs}	
\end{figure}
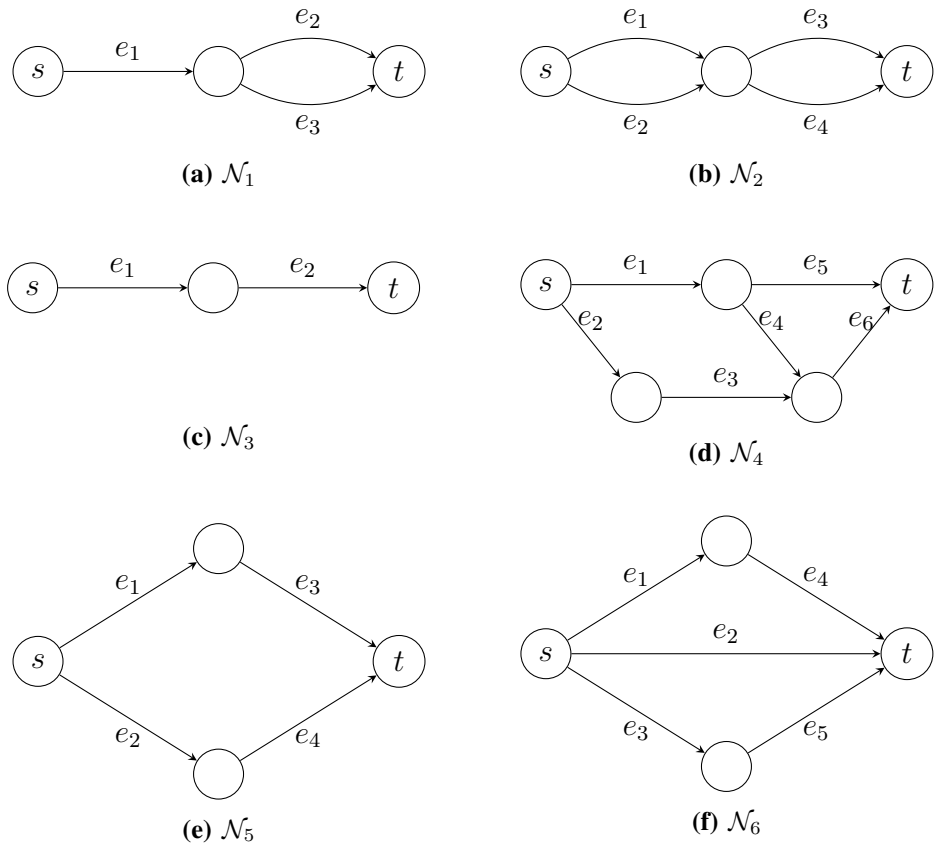
The computational complexity of the ergodic network capacity is made up of 
the enumeration of all minimal cut-sets, \eqref{eqn13}, \eqref{eqn14}, and \eqref{eqn15}. 
The enumeration of all minimal cut-sets has a complexity of 
$\mathcal{O}(|\mathcal{L}||\mathcal{V}|^3)$ as given in \cite{berry2000generating}. 
The total complexity of the functions defined in \eqref{eqn13}, \eqref{eqn14}, and 
\eqref{eqn15} is equal to $\mathcal{O}(m(|\mathcal{L}|+\zeta))$, 
where $\zeta=\sum\limits_{\mathcal{C}\in \mathcal{L}}^{} |\mathcal{C}|$. 
Hence the total computational complexity of the ergodic network capacity expression 
is equal to $\mathcal{O}\left(|\mathcal{L}||\mathcal{V}|^3+m(|\mathcal{L}|+\zeta)\right)$. 
Note that $k$ and $|\mathcal{L}|$ increase exponentially with the size of $\mathcal{N}$ 
\cite{provan1996p}.
\section{Numerical Results}	

In this section, we present numerical results to clarify the theoretical
expressions on the network performance. We provide two instructive examples to
clarify theoretical expressions derived in the previous sections.

Firstly, consider the example network $\mathcal{N}_1$ presented in Fig.
\ref{fig:all_graphs}(\subref{fig:N1}) with edges as labeled.
In this network, there are
$n=3$ edges with the size of the minimum cut $m=1$.
The cut-sets, minimal cut-sets, and minimum cut-sets are
\begin{align*}
\mathcal{K}  =& \{ \{ e_1 \}, \{ e_1,e_2 \}, \{ e_1,e_3 \}, 
                 \{ e_2,e_3 \}, \{ e_1,e_2,e_3 \} \},\\
\mathcal{L}  =& \{ \{ e_1 \}, \{ e_2,e_3 \}  \}, \text{ and } \\
\mathcal{M}  =& \{ \{ e_1 \}  \} , \text{ respectively}.
\end{align*}
We have $A(x) = x + 3x^2 + x^3$, thus, the outage polynomial for
 $\mathcal{N}_1$ is calculated as
\begin{equation*}
O(p) = p(1-p)^2 + 3p^2(1-p) + p^3 = p + p^2 - p^3.
\label{eqn30}
\end{equation*} 
The bound expressions are also given by 
\begin{align*}
O(p) \leq  & A(p) = p + 3p^2 + p^3 \\
O(p) \leq & p + p^2 \\
O(p) \geq & p(1-p)^2.
\end{align*}
\begin{figure}[t!]
	\centering
	\includegraphics[width=0.75\linewidth]{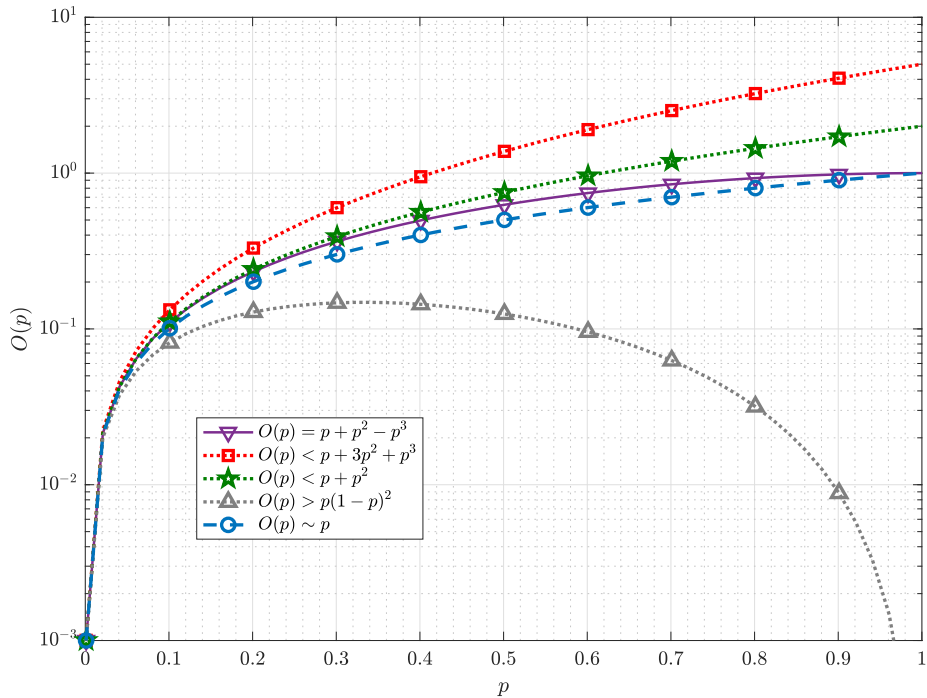}
	\caption{The comparative results of upper and lower bounds of 
		the outage polynomial of the $\mathcal{N}_1$.}
	\label{fig:outage_bounds_3edges}
\end{figure}
	
Fig.~\ref{fig:outage_bounds_3edges} shows that
the given upper
bounds become tight when $p>0.5$. As $p \rightarrow 0$, all bounds have the
same outage performance with the exact $O(p)$ expression. In addition, the
first order approximation of $O(p)$ given as $O(p)\sim p$ has a close performance
with $O(p)$ along with $p$.

Based on the given capacity assurance sets, capacity polynomials 
of the network can be calculated as:
\begin{align*} 
C_0(p) &=O(p)= p + p^2 - p^3   \\
C_1(p) &=2p(1-p)^2+(1-p)^3=1-p-p^2+p^3. 
\label{eqn33}
\end{align*}
By using \eqref{ergodic}, the ergodic capacity of $\mathcal{N}_1$ 
can be found as:
\begin{equation*}
E[C](p) =1-p-p^2+p^3. 
\label{eqn34}
\end{equation*}	
The obtained capacity polynomials of $\mathcal{N}_1$ are presented in
Fig.~\ref{fig:capacity_3edges} (\subref{fig:capacity_3edges_a}). While $p<0.5$,
$C_0(p)$ is highly probable when compared to $C_m(p)$ for $m=1$. On
the other hand, $C_1(p) \rightarrow 1$ in the case of $p \rightarrow 0$.
It can be deduced from Fig. \ref{fig:capacity_3edges}
(\subref{fig:capacity_3edges_b}),  the average capacity of the network
increases while $p$ is decreasing. In addition, the maximum value of the
average capacity of the network is equal to $m=1$ for $p=0$.

We give another example to illustrate the correlated case results. The depicted
extended graph of Fig. \ref{fig:all_graphs} (\subref{fig:N1}) with 4 edges
labeled as $\mathcal{N}_2$ is given in Fig. \ref{fig:all_graphs}
(\subref{fig:N2}). 

\begin{figure}[tb]
	\centering
	\begin{subfigure}[b]{0.45\textwidth}
		\includegraphics[width=\textwidth]{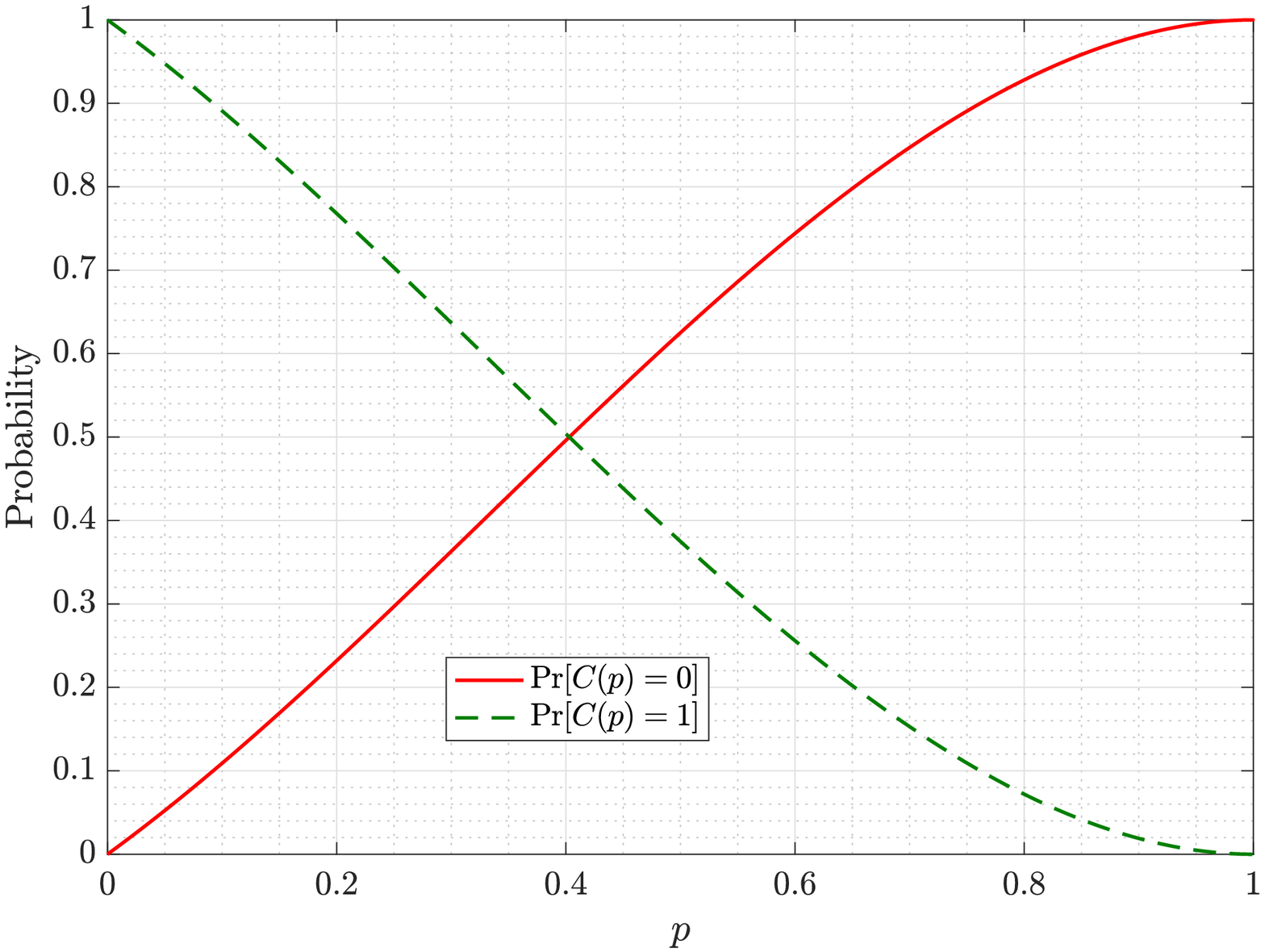}
		\caption{}
		\label{fig:capacity_3edges_a}
	\end{subfigure}
	\begin{subfigure}[b]{0.45\textwidth}
		\includegraphics[width=\textwidth]{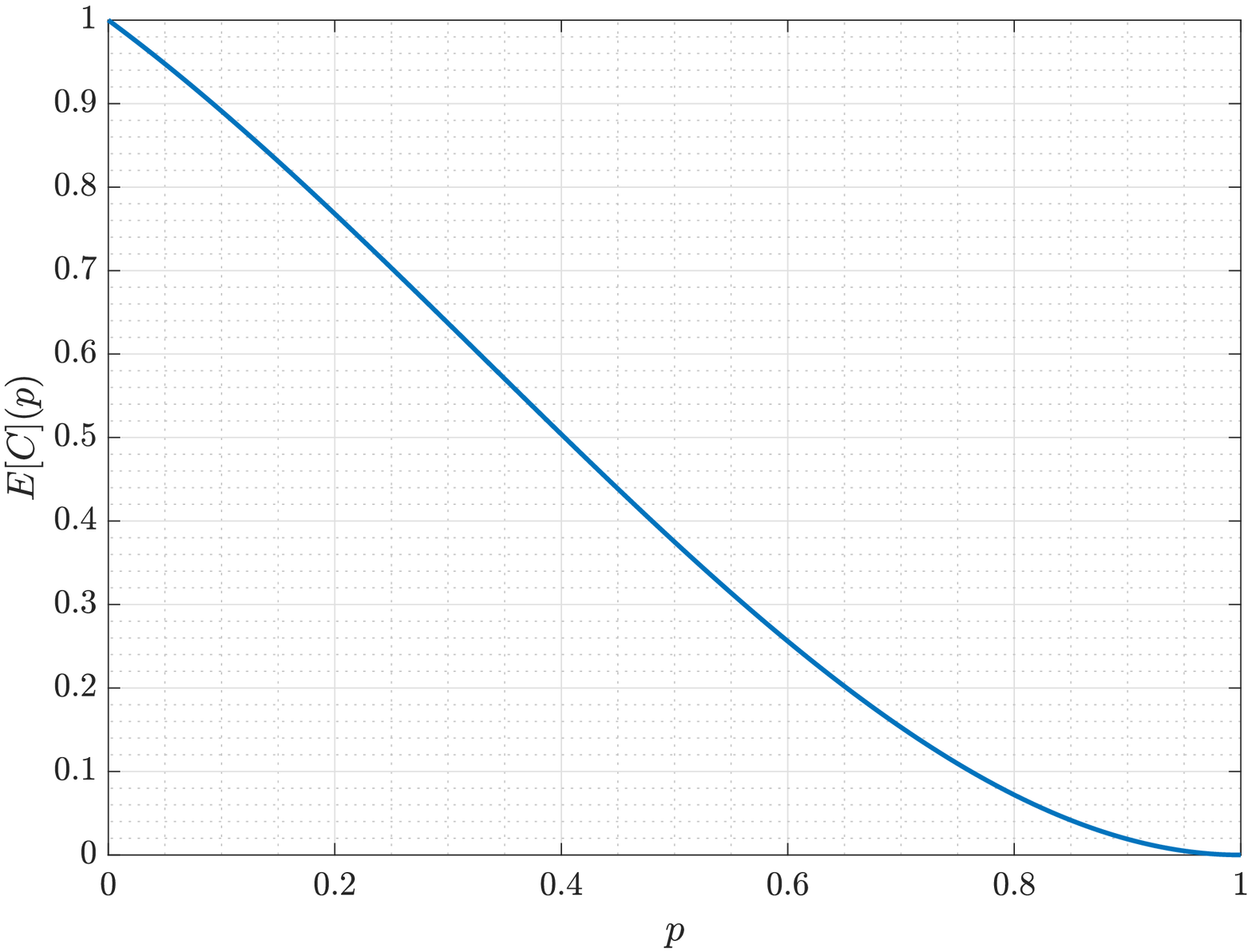}
		\caption{}
		\label{fig:capacity_3edges_b}
	\end{subfigure}	
	\caption{(a) The capacity polynomials results of $\mathcal{N}_1$ versus varying  
		$p$ value are presented. (b) The ergodic capacity results of $\mathcal{N}_1$.}
	\label{fig:capacity_3edges}
\end{figure}
The network has the following sets:
\begin{align*}
\mathcal{K}&= \{\{e_1,e_2\}, \{e_3,e_4\}, \{e_1,e_2,e_3\}, 
\{e_1,e_2,e_4\}, \{e_2,e_3,e_4\}, \{e_1,e_3,e_4\}, \{e_1,e_2,e_3,e_4\}\}  
\\ \mathcal{L} &= \{\{e_1,e_2\}, \{e_3,e_4 \}\} 
\\ \mathcal{M}&=\{\{ e_1,e_2\}, \{e_3,e_4 \}\} 
\label{eqn35}
\end{align*}
In the uncorrelated case, the outage polynomial can be calculated as:
\begin{equation*}
O(p) = 2p^2(1-p)^2 +4p^3(1-p)+p^4=2p^2-p^4
\label{eqn36}
\end{equation*}
where $m=2$ and $A_m=2$. If the correlated edge assumption given in \eqref{eqn:correlation} is used, the disjoint edge sets are given as
\begin{equation*}
\mathcal{B}_1=\{e_1,\; e_2\}, \; \mathcal{B}_2=\{e_3,\;e_4\}.
\end{equation*}  
where $\mathcal{B}_1 \cup \mathcal{B}_2=\mathcal{E} \;\textrm{and}\; \mathcal{B}_1 \cap \mathcal{B}_2=\emptyset$. By using \eqref{eqn:correlatedoutage}, the outage polynomial of the correlated case is derived as:
\begin{equation}
O(p)=(\rho p+p^2-\rho p^2)\left[2-\rho p-p^2 + \rho p^2 \right]
\label{eqn37}
\end{equation}
	
\begin{figure}[t]
	\centering
	\includegraphics[width=0.75\linewidth]{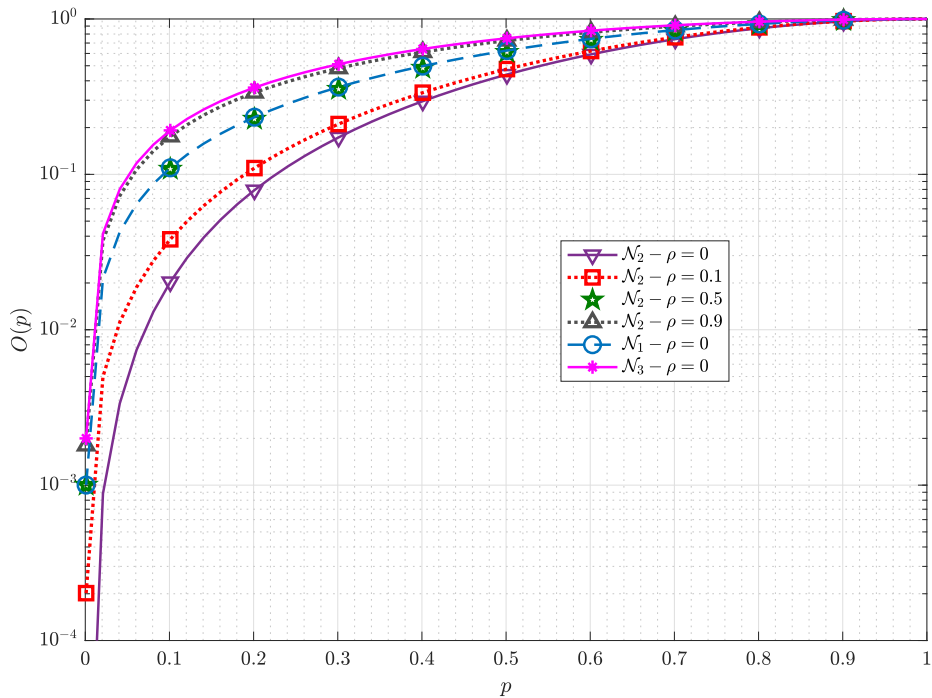}
	\caption{The outage polynomial results of $\mathcal{N}_2$ are presented in various correlation coefficient, $\rho$.}
	\label{fig:outage_4edges}
\end{figure}
The numerical results of \eqref{eqn37} are presented in Fig. \ref{fig:outage_4edges}. The uncorrelated case ($\rho=0$) has the best outage performance as expected. When $\rho=0.1$, the outage performance is worse than the uncorrelated case.  As $\rho$ increases, the outage performance gets worse. When $\rho=0.5$, 4-edges network has the close performance of 3-edges system handled in Example 1. On the other hand, if $\rho$ is equal to $0.9$ it means that highly correlated links are available, the outage performance of 4-edges networks approaches to 2-edges system model labeled as $\mathcal{N}_3$ is given in Fig. \ref{fig:all_graphs} (\subref{fig:N3}).
	
The capacity polynomials of $\mathcal{N}_2$ is given in Fig. \ref{fig:outage_4edges} can be calculated as:
\begin{align*}
C_0(p) &= p^4+4p^3(1-p) + 2p^2(1-p)^2 =2p^2-p^4    \\
C_1(p) &=4p^2(1-p)^2+4p(1-p)^3=4p-8p^2+4p^3   \\
C_2(p) &= (1-p)^4 . 
\label{eqn38}
\end{align*}
Hence, the ergodic capacity of  $\mathcal{N}_2$ is given by
\begin{equation*}
E[C](p)=2-4p+4p^2-4p^3+2p^4.
\label{eqn39}
\end{equation*}
The numerical results of the given polynomials are shown in Fig. \ref{fig:capacity_4edges} (\subref{fig:capacity_4edges_a}).
While $p<0.5$, ${C_m(p)}<m$ is high than ${C_m}(p)$. On the other hand, ${C_m}(p) \rightarrow 1$ in the case of $p \rightarrow 0$. Hence, the maximum value of the ergodic capacity of the network which is shown in Fig.  \ref{fig:capacity_4edges} (\subref{fig:capacity_4edges_b}) is equal to $m=2$ for $p \rightarrow 0$.
	\begin{figure}[t]
	\centering
	\begin{subfigure}[b]{0.45\textwidth}
		\includegraphics[width=\textwidth]{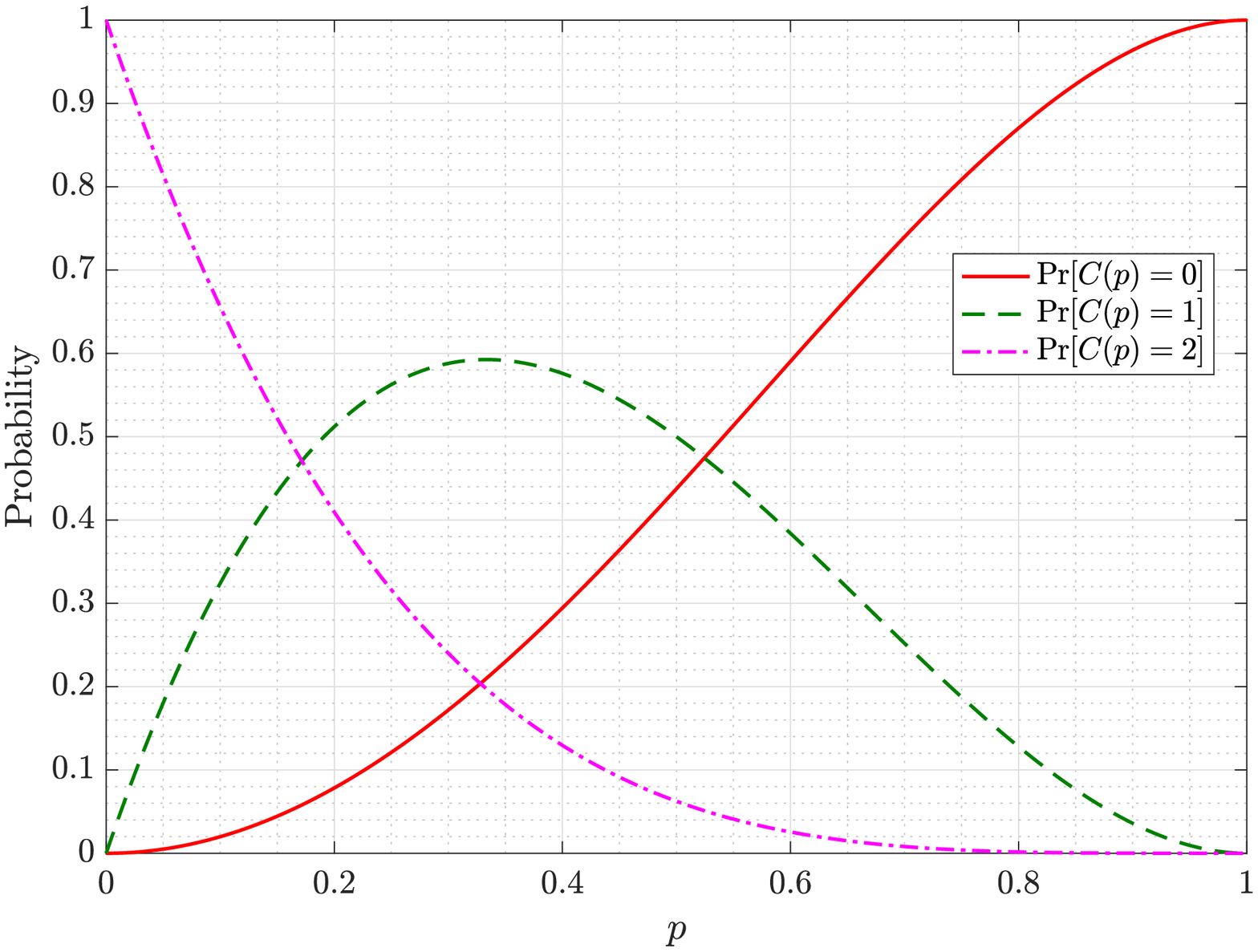}
		\caption{}
		\label{fig:capacity_4edges_a}
	\end{subfigure}
	\begin{subfigure}[b]{0.45\textwidth}
		\includegraphics[width=\textwidth]{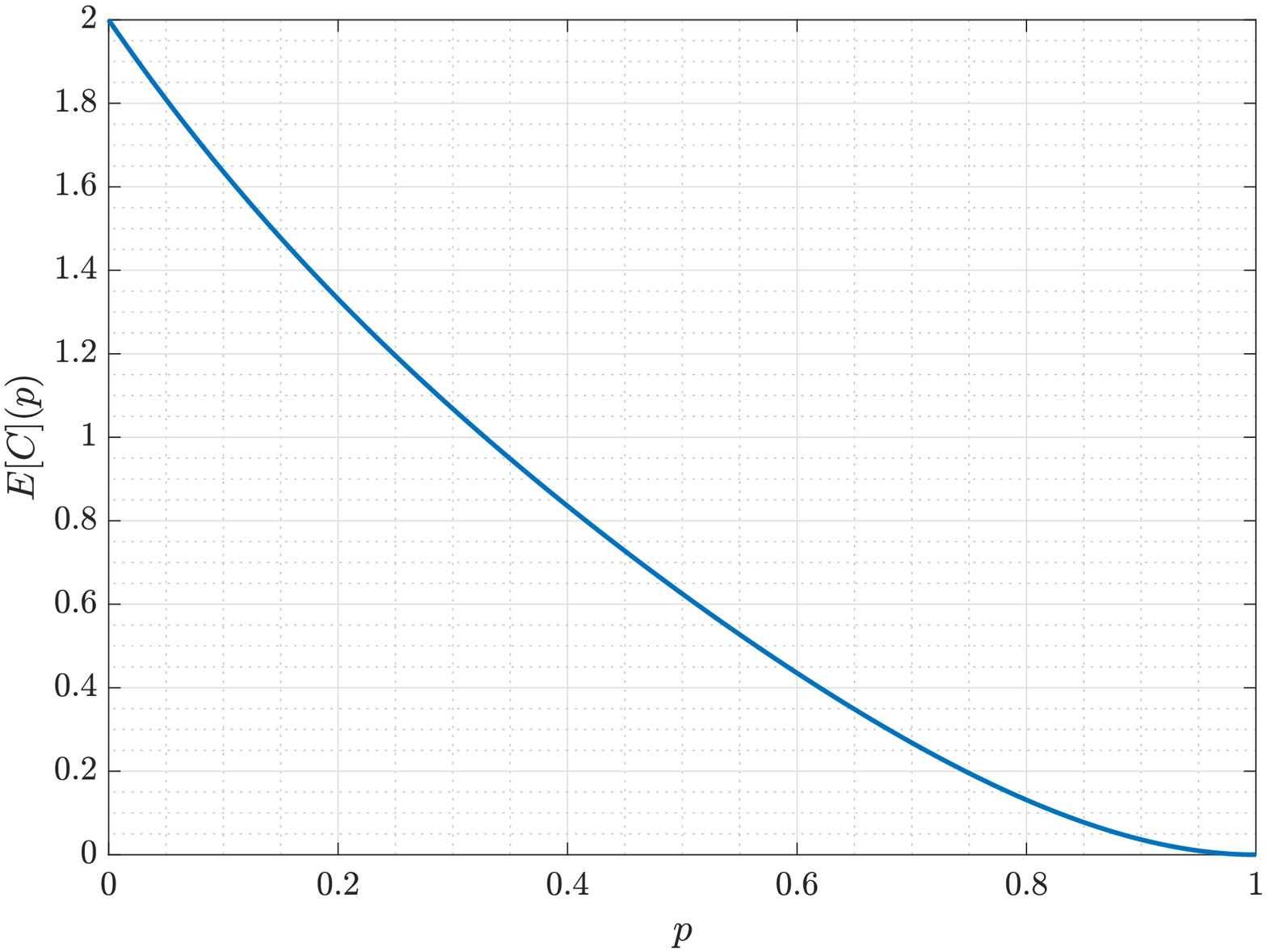}
		\caption{}
		\label{fig:capacity_4edges_b}
	\end{subfigure}	
	\caption{(a) The capacity polynomials results of the $\mathcal{N}_2$ versus varying  $p$ value are presented. (b) The ergodic capacity results of the network $\mathcal{N}_2$ .}
	\label{fig:capacity_4edges}
\end{figure}
\begin{figure}[t]
	\centering
	\begin{subfigure}[b]{0.49\textwidth}
		\includegraphics[width=\textwidth]{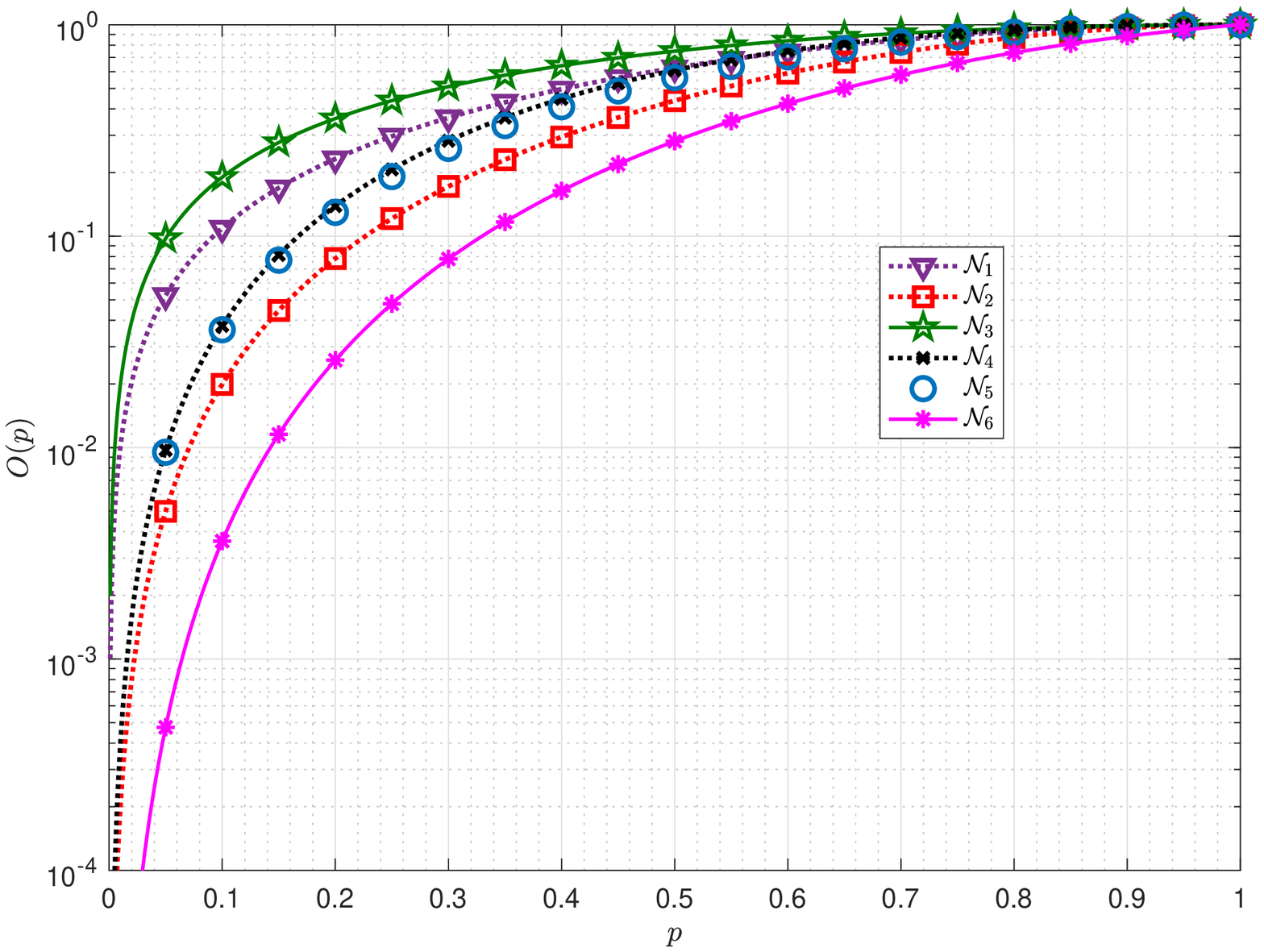}
		\caption{}
		\label{fig:outage_all_graphs}
	\end{subfigure}
	\begin{subfigure}[b]{0.49\textwidth}
		\includegraphics[width=\textwidth]{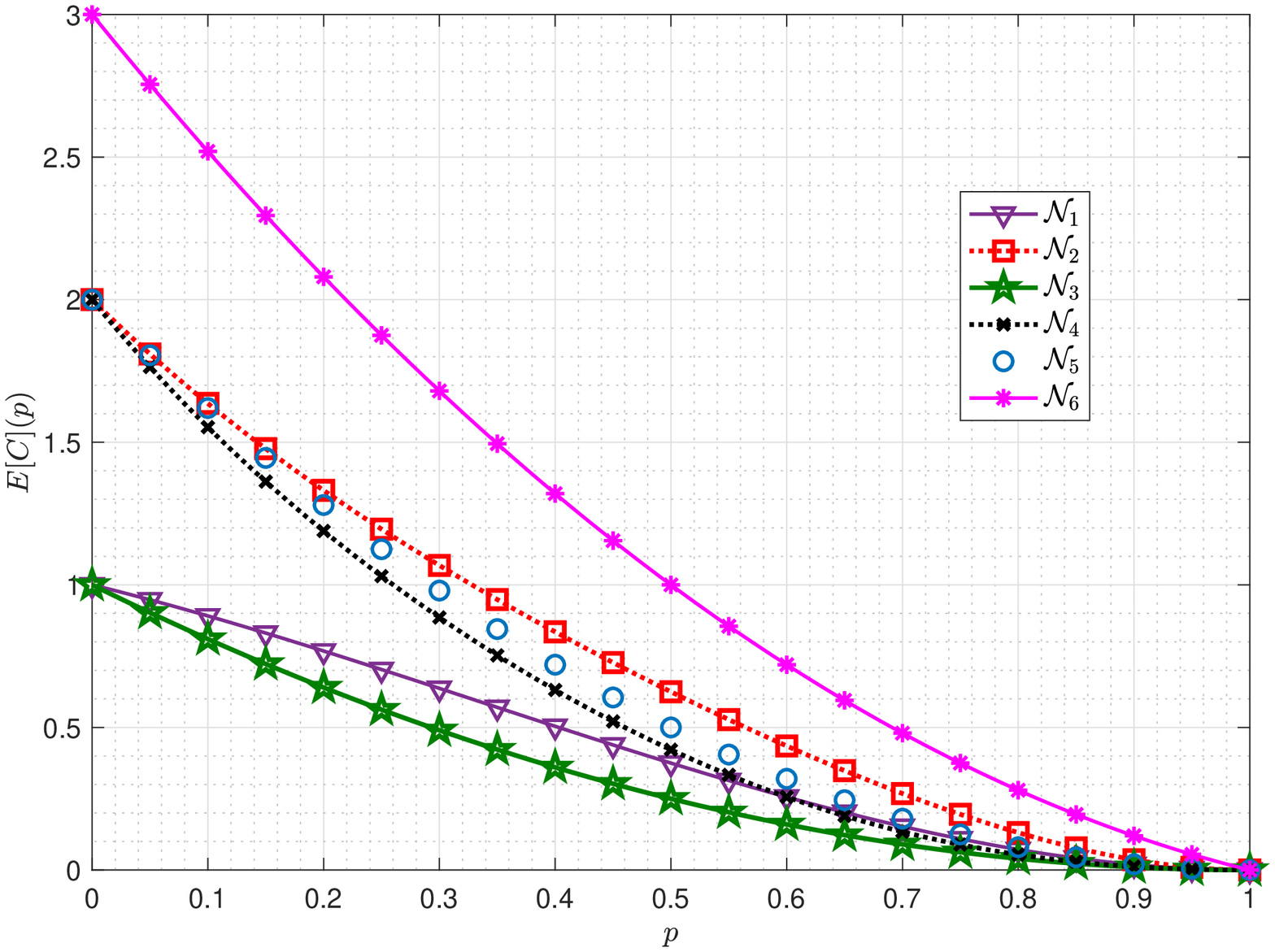}
		\caption{}
		\label{fig:capacity_all_graphs}
	\end{subfigure}	
	\caption{(a) The outage polynomial results of all the networks given in Fig. \ref{fig:all_graphs}. (b) The ergodic capacity results of all the networks given in Fig. \ref{fig:all_graphs}.}
	\label{fig:all_graphs_results}
\end{figure}

In order to obtain further insight about the derivations, the outage polynomial and the ergodic capacity results of $\mathcal{N}_4$, $\mathcal{N}_5$, and $\mathcal{N}_6$ depicted in Fig.s \ref{fig:all_graphs} (\subref{fig:N4}), (\subref{fig:N5}) and (\subref{fig:N6}), respectively, are investigated. By using \eqref{eqn16}, the outage polynomial expressions of $\mathcal{N}_4$, $\mathcal{N}_5$, and $\mathcal{N}_6$ can be respectively calculated as:
\begin{align*}
O(p)=&4p^2-2p^3-4p^4+4p^5-p^6,\\
O(p)=&4p^2-4p^3+p^4 \\
O(p)=&4p^3-4p^4+p^5, 
\end{align*}
Here, the three graphs have the same coding gain with $A_m=4$. On the other hand, $\mathcal{N}_4$ and $\mathcal{N}_5$ have the same diversity order equal to 2 and the diversity order of $\mathcal{N}_6$ is equal to 3. The ergodic capacity results of the three networks can be respectively given as:
\begin{align*}
E [C](p)=&2-5p+6p^2-8p^3+9p^4-5p^5+p^6,\\
E [C](p)=&2-4p+2p^2 \\
E [C](p)=&3-5p+2p^2 
\end{align*}
The outage polynomial and the ergodic capacity results of all the networks shown in Fig. \ref{fig:all_graphs} are presented in Fig. \ref{fig:all_graphs_results}. It can be deduced  from Fig. \ref{fig:all_graphs}(\subref{fig:outage_all_graphs}) that $\mathcal{N}_6$ has the best outage performance with the highest diversity order $m=3$. The two worst outage performance with $m=1$ belongs to $\mathcal{N}_3$ and $\mathcal{N}_1$, as expected. $\mathcal{N}_2$, $\mathcal{N}_4$, and $\mathcal{N}_5$ have close outage performance results with $m=2$. The ergodic capacity results are in accordance with the outage polynomial results. Hence, the best performance belongs to $\mathcal{N}_6$.

\section{Conclusion}

In this paper, we have obtained the performance limits of  generalized wireless
communication networks by using the concepts of graph theory. We have evaluated
the network outage polynomial by utilizing individual link outages, through the
use of path enumeration, cut-set enumeration and terminal-reliability
approaches. For high-SNR region, diversity order and coding gain have been
extracted from the graph model of wireless networks. We have proven that the
diversity order of any wireless  communication network is minimum cut-set size
of the network graph and the coding gain  is the number of distinct minimum
cut-sets. We have also presented the ergodic capacity analysis of arbitrary
networks to obtain the ergodic capacity polynomials.  The theoretical
expressions have been illustrated by numerical examples. Hence, we have
provided a comprehensive tool can be used to determine asymptotic performance
of unstructured wireless networks and to specify their performance limitations
under various implementation schemes.

\bibliographystyle{IEEEtran}
\bibliography{refs} 	
\end{document}